\documentclass[final]{l4dc2025}

\usepackage{bm}
\usepackage{color}
\usepackage[toc,page]{appendix}
\usepackage{wrapfig}
\usepackage{hyperref}
\usepackage{fp, tikz}
\usepackage{booktabs}
\usepackage{pgfplots}
\usepackage{caption}
\usepgfplotslibrary{fillbetween}
\pgfplotsset{compat=newest}
\definecolor{azrablue}{RGB}{0, 103, 180}
\definecolor{gr}{RGB}{85, 139, 47}
\definecolor{pr}{RGB}{128, 0, 128} 
\definecolor{or}{RGB}{255, 165, 0} 
\usepackage[font=small,labelfont=bf]{caption}

\newtheorem{assumption}{Assumption}

\allowdisplaybreaks

\newcommand{\AB}[1]{\textcolor{black}{#1}}
\newcommand{\DCH}[1]{\textcolor{black}{#1}}

\title[Back to Base: Learning using Safe Resets with Reach-Avoid Safety Filters]{Back to Base: Towards Hands-Off Learning via Safe Resets with Reach-Avoid Safety Filters}
\usepackage{times}
\usepackage{comment}

\author{%
 \Name{Azra Begzadi\'c$^{*}$} \Email{abegzadic@ucsd.edu} \vspace{-2.3mm}
 \AND
 \Name{Nikhil Uday Shinde$^{*}$} \Email{nshinde@ucsd.edu} \vspace{-2.3mm}
 \AND
  \Name{Sander Tonkens$^{*}$} \Email{stonkens@ucsd.edu} \vspace{-2.3mm}
 \AND
 \Name{Dylan Hirsch} \Email{dhirsch@uscd.edu} \vspace{-2.3mm}
 \AND
 \Name{Kaleb Ugalde} \Email{kugalde@ucsd.edu} \vspace{-2.3mm}
 \AND
 \Name{Michael C. Yip} \Email{yip@ucsd.edu} \vspace{-2.3mm}
 \AND
 \Name{Jorge Cortés} \Email{cortes@ucsd.edu} \vspace{-2.3mm}
 \AND
 \Name{Sylvia Herbert} \Email{sherbert@ucsd.edu} \vspace{-2mm}
 \AND
\addr{Contextual Robotics Institute, University of California San Diego \\ \vspace{-0.7mm} $^*$ These authors contributed equally to this work.}
}

\begin{document}
\setlength{\abovedisplayskip}{6pt}
\setlength{\belowdisplayskip}{6pt}
\maketitle
\begin{abstract}%
Designing controllers that accomplish tasks while guaranteeing safety constraints remains a significant challenge. We often want an agent to perform well in a nominal task, such as environment exploration, while ensuring it can avoid unsafe states and return to a desired target by a specific time. In particular we are motivated by the setting of safe, efficient, hands-off training for reinforcement learning in the real world. 
By enabling a robot to safely and autonomously reset to a desired region (e.g., charging stations) without human intervention, we can enhance efficiency and facilitate training.
Safety filters, such as those based on control barrier functions, decouple safety from nominal control objectives and rigorously guarantee safety. 
Despite their success, constructing these functions for general nonlinear systems with control constraints and system uncertainties remains an open problem. This paper introduces a safety filter obtained from the value function associated with the reach-avoid problem.  
The proposed safety filter minimally modifies the nominal controller while avoiding unsafe regions and guiding the system back to the desired target set. By preserving policy performance while allowing safe resetting, we enable efficient hands-off reinforcement learning and advance the feasibility of safe training for real world robots. We demonstrate our approach using a modified version of soft actor-critic to safely train a swing-up task on a modified cartpole stabilization problem.

\end{abstract}
\vspace{-1mm}
\section{Introduction}
Reach-avoid problems have typically focused on reaching a target set in the state space as quickly as possible while avoiding a failure set.
One popular value function-based approach is Hamilton-Jacobi reachability (HJR) analysis~\citep{Bansal2017HamiltonJacobiRA}.
HJR encodes a differential game between the control and disturbance to the system.  The result is a value function whose level sets define the reach-avoid set (or tube): the set of states from which the target can be reached safely despite worst-case disturbance at the end of (or within) the time horizon. Additionally, the gradients of the value function provide the optimal control strategy for reaching the goal \textit{in minimum time}.

A secondary approach is through the separate use of control Lyapunov functions (CLFs) that enforce stabilizing to a desired goal, combined with control barrier functions (CBFs) that enforce safety~\citep{Ames2017ControlBF}.  
An optimization problem (which is not necessarily feasible) incorporating both constraints is solved at each iteration, forcing the control to maintain safety (via the CBF) while progressing toward the goal (via the CLF).
However, the CLF constraint only ensures \textit{exponential stabilizability} to the goal. 

In practical applications, many scenarios require safely reaching a target at or within an exact time in the future, while achieving a secondary objective such as minimizing control input. 
One such application is hybrid systems, where the system must safely reach a switching state within a given time \citep{Grandia2020MultiLayeredSF}.
Additionally, a broad range of systems require being able to return to a desired goal within a certain time, e.g., a drone has to safely return to a dock station before its battery runs out.
In this paper, we motivate our work through the application of safely training a reinforcement learning (RL) agent in the real world, either from scratch or through fine-tuning. Real world RL training entails significant risks to both the robot and its environment if safety is not properly ensured. Furthermore, the trajectory of the robot can lead to states where it cannot continue operating due to system limitations, such as running out of battery or getting stuck~\citep{zhu2020ingredients}. This often necessitates human intervention to reset the robot to a desired state from which the robot can continue learning. This can make real world training prohibitively expensive and even infeasible in certain scenarios, like search and rescue or space exploration. For example, a drone fine-tuning its flight or search policy during deployment in an unknown environment could benefit from an approach that minimally modifies the nominal policy to allow for effective learning, while ensuring that it can always safely return to a desired target set. This would allow the drone to seamlessly resume its tasks and continue learning, enhancing the efficiency of real world RL.

We propose characterizing this problem through a time-varying reach-avoid value function.
However, instead of directly applying the optimal control to safely reach the goal in minimum time, as is typically done using HJR, we view the reach-avoid tube as a time-varying safe set and use an associated CBF-like constraint on the value function to design a safety filter. This safety filter allows a user to prioritize safety while also pursuing a secondary objective (e.g. learning a performance policy using RL). Our contributions are as follows: \vspace{-2mm}
\begin{enumerate}
    \item We introduce the notion of a viscosity-based control barrier function (VB-CBF), which resembles a standard CBF but has less restrictive assumptions. For control-and disturbance affine dynamics its associated time-varying safety filter is a quadratic program.\vspace{-3mm}
    \item We prove that the HJR reach-avoid value function is a VB-CBF for the reach-avoid set, and admits a feasible control set for almost every state and time pair in the reach-avoid set. Under mild assumptions, we show this value function is also a VB-CBF for the reach-avoid tube. \vspace{-3mm}
    \item We demonstrate the effectiveness of our approach for safely training an RL agent on an enhanced cartpole environment. We successfully showcase how we can use our method to train an RL agent without the need to explicitly reset the environment. Specifically, the reach-avoid VB-CBF enables both training safely and returning to a safe initial state at the end of each episode. Both are crucial components for scalable real world RL.  \vspace{-1mm}
\end{enumerate}
\section{Related Work}
\textbf{Hamilton-Jacobi Reachability.} Many safety-critical autonomous systems tasks can be formulated as reach-avoid problems~\citep{MargellosReachAvoid, fisac2015reach}. HJR provides a rigorous framework for verifying safety and reachability in dynamical systems~\citep{Bansal2017HamiltonJacobiRA}. \AB{To address the scalability challenges of HJR, learning-based approaches have been developed~\citep{BansalTomlin2021,HsuRubies-RoyoEtAl2021, Chnevert2024SolvingRP, Chung2024GuaranteedRF}.}
However, these methods require specifying the desired task objectives as a part of their reach-avoid formulation \emph{a priori}, and are hence limited to a specific subset of tasks.
In contrast,~\citep{Yu2022ReachabilityCR} uses HJR to define a feasible set and constrain an RL agent to learn an unspecified objective constrained to this set. 
This method considers constraints only and does not enforce reaching a target region.\\
\textbf{Control barrier functions.} Safety is also often ensured through the use of CBFs~\citep{Ames2017ControlBF}. In particular, CBFs are used as safety filters by adjusting a nominal control law to ensure the system satisfies safety constraints~\citep{Wabersich2023DataDrivenSF}, resulting in a quadratic program. However, constructing a CBF and ensuring the feasibility of a CBF-based optimization problem is challenging~\citep{PM-JC:24-auto}. To tackle feasibility challenges, backup CBFs have been proposed to guarantee feasibility using a predefined backup policy to a predefined safe set~\citep{GurrietMoteEtAl2020, ChenJankovicEtAl2021}. While more broadly applicable to complex systems, specifying a backup set a priori limits the implicitly specified safety region and the backup CBF formulation introduces additional computational complexity. Deriving a CBF from HJR without a target (i.e., only avoiding unsafe states) is explored in~\cite{ChoiLeeEtAl2021} via a novel value function. Similarly, in~\cite{SoSerlinEtAl2024}, the authors provide a method for constructing CBFs using ideas from reachability analysis, but their approach is limited to fixed policies. \\
\textbf{Safe \& recoverable RL.} Safe RL is well-explored~\citep{brunke2022safe} and can be broadly categorized into framing the problem as a constrained Markov decision process~\citep{Altman} or using control-theoretic methods to restrict the action space of the agent.  
These works leverage that separating task performance from safety objectives can improve both performance and safety~\citep{ThananjeyanBalakrishnaEtAl2021}. Lyapunov-based methods, including CBFs, have been used to guarantee safety while learning for model-based~\citep{BergenkampTurchettaEtAl2017, Choi2020ReinforcementLF} and model-free RL~\citep{Cheng2019EndtoEndSR}. This has been extended to include system uncertainties and disturbances~\citep{Yousef2022, Cheng2022SafeAE}.  
However, these methods typically do not consider resetting to a desired state at the end of each episode, which is desired for autonomous, i.e., without human presence, training. In contrast, reach-avoid methods focus exclusively on reaching a desired goal state, which is often difficult to define in advance and limits the use of general-purpose RL algorithms. To address these challenges, we propose a time-varying CBF that encodes both safety (for all times) and recovery to a desired target set within the specified time.
\vspace{-1mm}
\section{Preliminaries} \label{sec:problemstatement}
Consider a system of the form
\begin{equation} \label{eq:sys}
    \dot{x} = f(x,u,d)
\end{equation}
where $x \in \mathbb{R}^n$ is the state, $u \in \mathcal{U} \subseteq \mathbb{R}^p$ is the control input, $d \in \mathcal{D} \subseteq \mathbb{R}^q$ is the disturbance, with convex and compact sets $\mathcal{U}$ and $\mathcal{D}$.  
For each initial time $t \le 0$, we denote the sets of admissible control and disturbance signals by $\mathbb{U}(t) := \{\bm{u} : [t, 0] \to \mathcal{U} \mid \bm{u} \text{ is measurable}\}$ and $\mathbb{D}(t):=\{\bm{d} : [t, 0] \to \mathcal{D} \mid \bm{d} \text{ is measurable}\}$, respectively.
Throughout this work, we make the following assumption on the dynamics.
\begin{assumption} \label{as:flowfield}
    The function $f: \mathbb{R}^n \times \mathcal{U} \times \mathcal{D} \to \mathbb{R}^n$ is bounded and globally Lipschitz.
\end{assumption}
It follows from Assumption \ref{as:flowfield} that for each $t \le 0$, $x \in \mathbb{R}^n$, $\bm{u} \in \mathbb{U}(t)$, and $\bm{d} \in \mathbb{D}(t)$, there exists a unique (Carath\'{e}odory) solution $\bm{x}:[t,0] \to \mathbb{R}^n$ of \eqref{eq:sys} which satisfies $\bm{x}(t) = x$~\citep{Coddington1955TheoryOO}. We denote this solution by $\xi_{x, t}^{\bm{u}, \bm{d}}$. \vspace{-1mm}
\subsection{Hamilton-Jacobi Reachability}
\AB{HJR determines the set of initial states from which a system can robustly reach a goal while avoiding failure states.} This analysis is formulated in terms of a differential game played over the dynamics \eqref{eq:sys}, where we consider the control $u$ and disturbance $d$ as the players of the game \citep{Bansal2017HamiltonJacobiRA}.
Player $u$ wishes to ensure the system enters the target set $\mathcal{T}\subseteq\mathbb{R}^n$ by some final time (which we shall choose to be $0$), while avoiding the failure set $\mathcal{F}\subseteq{R}^n$ in the process, and player $d$ wishes for the opposite.
Given an initial time $t \le 0$, the reach-avoid tube $\mathcal{RA} (\mathcal{T}\! , \mathcal{F}\! , t)$ represents the set of initial states $x$ for which the control can win the game.
More precisely, we have
\begin{equation}\label{eq:reach_avoid_set}
    \hspace{-.2cm}\mathcal{RA} (\!\mathcal{T}\!, \mathcal{F}\! , t\!)  \! :=\!  \{ \!  x \!  \in \!  \mathbb{R}^n \!  \mid \!   \forall \lambda \! \in  \!  \!  \Lambda(t)\!, \exists \mathbf{u} \! \in \! \mathbb{U}(t), \exists s \! \in \!  [t,\! 0], \xi_{x, t}^{\bm{u}, \lambda[\bm{u}]}(\!s\!) \!  \in \!  \mathcal{T} \!\land \!\forall  \tau \!  \in \! [t, \! s], \xi_{x, t}^{\bm{u}, \lambda[\bm{u}]}(\!\tau\!) \!  \notin \!  \mathcal{F} \! \},\!
\end{equation}
where $\Lambda(t)$ represents the set of non-anticipative (i.e., causal) strategies from which we permit the disturbance player to choose. Formally, we have
\begin{align}
    \Lambda(t) \! =\!  \{&\lambda:\mathbb{U}(t) \to \mathbb{D}(t) \! \mid \! \forall s \in [t,0], \forall \bm{u}, \hat{\bm{u}} \in \mathbb{U}(t),\nonumber\\
    &\bm{u}(\tau) \! = \!  \hat{\bm{u}}(\tau) \text{ a.e. } \tau \in [t,s] \implies \lambda[\bm{u}](\tau) \!  = \!  \lambda[\hat{\bm{u}}](\tau)\text{ a.e. } \tau \in [t,s]\} .
\end{align}
Though we are mainly interested in the tube, it will be helpful to consider the reach avoid set
\begin{equation}\label{eq:reach_avoid_set_t}
    \mathcal{RA}_0 (\mathcal{T}\! , \mathcal{F}\! , t)  \! :=\!  \{ x \!  \in \!  \mathbb{R}^n \!  \mid \! \forall \lambda \! \in \!  \Lambda(t),\exists \mathbf{u} \! \in \! \mathbb{U}(t),\xi_{x, t}^{\bm{u}, \lambda[\bm{u}]}(0) \!  \in \!  \mathcal{T} \land \forall \tau \! \in \! [t, 0], \xi_{x, t}^{\bm{u}, \lambda[\bm{u}]}(\tau) \!  \notin \!  \mathcal{F} \},
\end{equation}
which corresponds to a game in which the controller wishes to ensure the state is in the target set $\textit{at}$ the final time rather than $\textit{by}$ the final time. By definition, $\mathcal{RA}_0 (\mathcal{T}\! , \mathcal{F}\! , t)\subseteq\mathcal{RA} (\mathcal{T}\! , \mathcal{F}\! , t)$. To compute $\mathcal{RA}_0 (\mathcal{T}\! , \mathcal{F}\! , t)$, one first computes the value function $V_0: \mathbb{R}^n \times (-\infty,0] \to \mathbb{R}$ associated with the game. We define the target and constraint functions $\ell(x),  g(x):\mathbb{R}^n\to\mathbb{R}$ such that $\ell(x) \geq 0 \iff x \in \mathcal{T}$ and $g(x) < 0 \iff x\in \mathcal{F}$. Then, $V_0$ is the unique viscosity solution \DCH{(for details on viscosity solutions, we refer the reader to \cite{bardi})} of the following variational inequality \citep{MargellosReachAvoid, fisac2015reach}
\begin{equation}\label{eq:reach_avoid_VI_fixed_t}
\begin{cases}
    0 = \min\{g(x) - V_0(x,t), \frac{\partial}{\partial t} V_0(x,t) + H(\nabla V_0(x,t), x)\} \quad \text{ in } x \in \mathbb{R}^n, t<0\\
    V_0(x,0)=\min\{\ell(x), g(x)\} \quad \text{ on } x \in \mathbb{R}^n,
\end{cases}
\end{equation}
where the Hamiltonian $H: \mathbb{R}^n \times \mathbb{R}^n \to \mathbb{R}$ is given by $H(\lambda,x) = \max_{u \in \mathcal{U}}\min_{d \in \mathcal{D}}\lambda^\top f(x,u,d)$.
Then, we have the following relationship
\begin{equation}\label{eq:level-set-theorem}
    V_0(x,t) \geq 0 \iff x \in \mathcal{RA}_0(\mathcal{T}, \mathcal{F}, t).
\end{equation}
Moreover, the gradient of the value function $V_0(x,t)$ informs the optimal control law $u^*(x,t)$ 
\begin{equation}
    u^*(x,t) = \arg\max_{u\in\mathcal{U}} \min_{d\in\mathcal{D}} \nabla V_0(x,t)^\top f(x,u,d).
\end{equation}
Following this control law ensures that any state $x$ at time $t$ such that $V_0(x,t)\geq 0$ reaches the target while avoiding the failure set, despite the best effort from player~$d$.
\vspace{-1mm}
\subsection{Control Barrier Functions}
For consistency with the original work \citep{AmesXuEtAl2016}, we introduce CBFs in the context of control-affine dynamics with no disturbances, i.e.
\begin{equation} \label{eq:sys_cbf}
    \dot{x} = f_0(x) + g_0(x)u,
\end{equation}
where the functions $f_0 \!:\! \mathbb{R}^n \! \to \! \mathbb{R}^n$ and $g_0 \!:\! \mathbb{R}^n \! \to \! \mathbb{R}^{n \times p}$ are globally Lipschitz. 
Note that this formalism has been extended to more general dynamics, such as \eqref{eq:sys}, with the key ideas unchanged~\citep{Kolathaya2018InputtoStateSW}. Moreover, the system is considered safe as long as the state remains within a safe set $\mathcal{C} \subseteq \mathbb{R}^n$, defined as the zero-superlevel set of a continuously differentiable function $h : \mathbb{R}^n\to \mathbb{R}$. To ensure the safety, we introduce concept of CBFs, as defined in the following.
\begin{definition} \label{def:cbf}
    (\textbf{Control Barrier Functions}, \cite{AmesXuEtAl2016}) 
    A continuously differentiable function $h: \mathbb{R}^n \to \mathbb{R}$ such that $\mathcal{C} = \{x \in \mathbb{R}^n \mid h(x) \ge 0\}$
    is a control barrier function for \eqref{eq:sys_cbf} on $\mathcal{C}$ if there is an extended class $\mathcal{K}$ function\footnote{A function $\alpha: \mathbb{R} \to \mathbb{R}$ is said to be extended class $\mathcal{K}$ if $\alpha$ is continuous, strictly increasing, and satisfies $\alpha(0) = 0$.} $\alpha$ such that, for each $x\in \mathcal{C}$, there exists a control $u \in \mathcal{U}$ satisfying
    \begin{equation} \label{eq:cbf_gp}
        \nabla_{x} h(x)^\top (f_0(x)+g_0(x)u)+\alpha(h(x)) \geq 0.
    \end{equation}
\end{definition}
If one can identify a CBF for a system, the following result provides the desired safety guarantee:
\begin{theorem}(\cite{AmesXuEtAl2016}) If $h$ is a CBF for \eqref{eq:sys_cbf} on $\mathcal{C}$, and if $\nabla h(x) \ne 0$ for all $x \in \partial C$, then the set $\mathcal{C}$ is safe under any globally Lipschitz controller $u:\mathbb{R}^n \to \mathcal{U}$ for which \eqref{eq:cbf_gp} is satisfied with $u = u(x)$ at each $x \in \mathbb{R}^n$.
\end{theorem}
Given any nominal control law $u_\text{nom}: \mathbb{R}^n \to \mathbb{R}^p$ that might violate control limits and safety, CBFs can be used to minimally adjust the nominal control input with the following optimization problem: 
\begin{equation}\label{eq:qp}
\begin{aligned}
u^*(x)= &\arg \min _{u \in \mathcal{U}}\left\|u - u_{\text {nom }}(x)\right\|_2^2 \\
&\text { s.t. } \nabla_{x} h(x)^\top (f_0(x)+g_0(x)u)+\alpha(h(x)) \geq 0.
\end{aligned}
\end{equation}
Note that because the dynamics \eqref{eq:sys_cbf} are control-affine, then equation \eqref{eq:qp} is a quadratic program, enabling real-time safety filtering.
\vspace{-1mm}
\section{Safe \& Reset-Friendly Learning via Viscosity-Based Control Barrier Functions} \label{sec:main_method}
Joint safety, defined as avoiding failure states, and liveness, defined as reaching a target, are  commonly addressed through reach-avoid HJR or combined CLF-CBF approaches. While HJR ensures safety and liveness, it lacks a framework for minimally modifying a nominal controller in a smooth CBF-like fashion. In contrast, the CLF-CBF approach facilitates the construction of a safety filter, but ensuring feasibility becomes challenging in the presence of disturbances and control bounds.
As such, both of these methods are ill-suited to prioritizing finite-time safety and liveness for a system while executing a nominal controller for a secondary objective. 
Such a need arises, for example, when one wishes for a system to explore an environment while ensuring it can avoid failure states and return to a safe reset position (e.g., a charging station) within some time.

In this section, we develop a framework that robustly and safely reaches a goal within a desired time, while optimizing online for different performance objectives over the time horizon. In Section~\ref{sec:VB-CBF_set}, we develop a generalization of a CBF for reach-avoid sets that addresses safety and liveness under disturbances. The extension of these results to reach-avoid tubes is presented in Section~\ref{sec:VB-CBF_tube}. The design of a safety filter using our findings is covered in Section~\ref{sec:safety-filter}.
\vspace{-1mm}
\subsection{Viscosity-Based Control Barrier Functions for Reach-Avoid Sets} \label{sec:VB-CBF_set}
We consider a time-varying safe set $\mathcal{C}_v: (-\infty, 0] \to 2^{\mathbb{R}^n}$ that captures both finite-time safety and liveness. To achieve this objective, we introduce viscosity-based CBF (VB-CBFs) as follows.
\begin{definition}(\textbf{Viscosity-Based Control Barrier Function}) \label{def:ra-cbvf}
Consider a continuous function $h_v: \mathbb{R}^n \times (-\infty, 0] \to \mathbb{R}$, and for each $t \le 0$, let $\mathcal{C}_v(t) = \{ x \in \mathbb{R}^n \mid h_v(x,t) \ge 0 \}$. Then $h_v$ is a viscosity-based control barrier function (VB-CBF) for system \eqref{eq:sys} on $\mathcal{C}_v(\cdot)$ if there exists an extended class $\mathcal{K}$ function $\alpha$ such that for all $t < 0$ and all $x \in \mathcal{C}_v(t)$, the inequality
    \begin{equation} \label{eq:vb-cbf}
       \frac{\partial}{\partial t} h_v(x,t) + \max_{u \in \mathcal{U}} \min_{d \in \mathcal{D}} \nabla_x h_v(x, t)^\top f(x, u, d)  \geq - \alpha(h_v(x, t))
    \end{equation}
holds in a viscosity sense, i.e., for each continuously differentiable $\psi: \mathbb{R}^n \times (-\infty,0) \to \mathbb{R}$
\begin{equation} \label{eq:vb-cbf_psi}
       \frac{\partial}{\partial t} \psi(x,t) + \max_{u \in \mathcal{U}} \min_{d \in \mathcal{D}} \nabla_x \psi(x, t)^\top f(x, u, d)  \geq - \alpha(h_v(x, t))
    \end{equation}
 at each $(x,t)$ where $h_v - \psi$ has a local minimum.
\end{definition}
Unlike a traditional CBF, which satisfies the global safety condition~\eqref{eq:cbf_gp}, a VB-CBF satisfies the local safety condition \eqref{eq:vb-cbf_psi} anywhere such a $\psi$ exists (in particular anywhere $h_v$ is differentiable; see Lemma 1.7 in \cite{bardi}). Recall that the value function $V_0$ is the viscosity solution to the variational inequality \eqref{eq:reach_avoid_VI_fixed_t}.
The following result formalizes the connection between VB-CBF from Definition \ref{def:ra-cbvf} and the time-varying reach-avoid safety problem~\eqref{eq:reach_avoid_set_t}.
\begin{theorem} \label{th:main}
The value function $V_0:\mathbb{R}^n \times (-\infty,0] \to \mathbb{R}$ is a VB-CBF for system \eqref{eq:sys} on the reach-avoid set $\mathcal{RA}_0(\mathcal{T},\mathcal{F},\cdot)$.
\end{theorem}
\begin{proof}
 First, observe that by \eqref{eq:level-set-theorem}, $\mathcal{RA}_0(\mathcal{T},\mathcal{F},t)$ is indeed the zero-superlevel set of $V_0(\cdot,t)$ for each $t \le 0$.
 Fix $t < 0$ and $x \in \mathcal{RA}_0(\mathcal{T},\mathcal{F},t)$.
 Let $\psi:\mathbb{R}^n \to \mathbb{R}$ be a continuously differentiable function such that $V_0 - \psi$ has a local minimum at $(x,t)$.
 Since $V_0$ is a viscosity solution of \eqref{eq:reach_avoid_VI_fixed_t}, then
\begin{equation*}
    0 \leq \min\{g(x) - V_0(x,t), \frac{\partial}{\partial t} \psi(x,t) + \max_{u \in \mathcal{U}} \min_{d \in \mathcal{D}} \nabla \psi(x,t)^\top f(x,u,d)\}.
\end{equation*}
In particular, we obtain
\begin{equation*}
        \frac{\partial}{\partial t} \psi(x,t) + \max_{u \in \mathcal{U}} \min_{d \in \mathcal{D}}\nabla \psi(x,t)^\top f(x,u,d) \geq  0 \geq -\alpha (V_0(x,t)),
\end{equation*}
for any extended class $\mathcal{K}$ function $\alpha$, where the second inequality holds because $V_0(x,t) \ge 0$.
\end{proof}
\DCH{The above theorem provides a useful method for constructing a VB-CBF on the reach-avoid set, namely computing the value function $V_0$.}
\vspace{-1mm}
\subsection{Viscosity-Based Control Barrier Functions for Reach-Avoid Tubes} \label{sec:VB-CBF_tube}
While Theorem~\ref{th:main} provides a VB-CBF for the reach-avoid set $\mathcal{RA}_0(\mathcal{T}, \mathcal{F}, t)$ (i.e., safely reaching the target \emph{at} the final time), it is often desirable to safely reach the target \emph{by} the final time.
Therefore, in this section, we extend our result from previous section to reach-avoid tubes $\mathcal{RA}(\mathcal{T}, \mathcal{F}, t)$.
\begin{definition}(\textbf{Robust Control Invariance})
    A set $\mathcal{S} \subseteq \mathbb{R}^n$ is robustly control invariant if for each $t \le 0$ and $x \in \mathcal{S}$, for all $\lambda \in \Lambda(t)$ there exists a $\bm{u} \in \mathbb{U}(t)$ such that $\xi_{x,t}^{\bm{u},\lambda[\bm{u}]}(s) \in \mathcal{S}$ for all $s \in [t,0]$.
\end{definition}
For simplicity, we will assume one can safely stay within the target set and outside the failure set once the target has been reached. 
\begin{assumption}\label{as:rcitarget}
    The set $\mathcal{T} \setminus \mathcal{F}$ is non-empty and robustly control invariant.
\end{assumption}
Note that this assumption will usually be satisfied for our intended usage, where the target serves as a resetting set (e.g., a docking station). Therefore, Assumption \ref{as:rcitarget} is not restrictive in practice.
\begin{proposition}\label{prop:tube_set}
    Under Assumption~\ref{as:rcitarget}, the fixed-time reach-avoid set is the same as the reach-avoid tube, i.e. $\mathcal{RA}_0(\mathcal{T}, \mathcal{F}, t)=\mathcal{RA}(\mathcal{T}, \mathcal{F}, t)$ for each $t \le 0$.
\end{proposition}
\begin{proof}
     Fix $t < 0$ (the proof is trivial when $t = 0$).
     It is clear that $\mathcal{RA}_0(\mathcal{T}, \mathcal{F}, t) \subseteq \mathcal{RA}(\mathcal{T}, \mathcal{F}, t)$.
     We prove the reverse inclusion.
     Suppose $x \in \mathcal{RA}(\mathcal{T}, \mathcal{F}, t)$.
     Let $\lambda \in \Lambda(t)$.
     By the definition of $\mathcal{RA}(\mathcal{T}, \mathcal{F}, t)$, we can choose $\bm{u} \in \mathbb{U}(t)$ and $s \in [t,0]$ such that $\xi_{x,t}^{\bm{u},\bm{d}}(s) \in \mathcal{T}$ and $\xi_{x,t}^{\bm{u},\bm{d}}(\tau) \notin \mathcal{F}$ for all $\tau \in [t,s]$, where $\bm{d} := \lambda[\bm{u}]$. Let $\bm{u}_s = \bm{u}|[t,s]$ (i.e. the restriction of $\bm{u}$ to $[t,s]$), let $\bm{d}_s = \bm{d}|[t,s]$, and let $\lambda_s:\mathbb{U}(s) \to \mathbb{D}(s)$ be given by $\lambda_s: \bm{u}_0 \mapsto \lambda(\langle\bm{u_s},\bm{u_0}\rangle)|[s,0]$,
     where $\langle\bm{u_s},\bm{u_0}\rangle$ represents the concatenation of $\bm{u_s}$ and $\bm{u_0}$.
     It can be readily checked that non-anticipativity of $\lambda_s$ follows from non-anticipativity of $\lambda$. Let $x_s = \xi_{x,t}^{\bm{u},\bm{d}}(s)$.
     By Assumption \ref{as:rcitarget}, we can choose $\bm{u}_0^* \in \mathbb{U}(s)$ such that $\xi_{x_s,s}^{\bm{u}_0^*, \bm{d}_0^*}(\tau) \in \mathcal{T} \setminus \mathcal{F}$ for all $\tau \in [s,0]$, where $\bm{d}_0^* := \lambda_s[\bm{u}_0^*]$.
     Letting $\bm{u}^* = \langle \bm{u}_s,\bm{u}_0^* \rangle$ and $\bm{d}^* = \langle \bm{d}_s,\bm{d}_0^* \rangle$, we then have by non-anticipativity of $\lambda$ and the definition of $\lambda_s$ that $\lambda[\bm{u}^*] = \bm{d}^*$. Thus $\xi_{x,t}^{\bm{u}^*, \bm{d}^*}(\tau) = \xi_{x,t}^{\bm{u}, \bm{d}}(\tau)\notin \mathcal{F}$ for all $\tau \in [t,s]$. Moreover, $\xi_{x,t}^{\bm{u}^*, \bm{d}^*}(\tau) = \xi_{x_s,s}^{\bm{u}^*_0, \bm{d}^*_0}(\tau) \in \mathcal{T} \setminus \mathcal{F}$ for all $\tau \in [s,0]$.
     In other words, $\bm{u}^* \in \mathbb{U}(t)$ is such that $\xi_{x,t}^{\bm{u}^*, \lambda[\bm{u}^*]}(0) \in \mathcal{T}$ and $\xi_{x,t}^{\bm{u}^*, \lambda[\bm{u}^*]}(\tau) \notin \mathcal{F}$ for all $\tau \in [t,0]$.
     Thus $x \in \mathcal{RA}_0(\mathcal{T},\mathcal{F},t)$, completing the proof.
 \end{proof}
Intuitively, this proposition states that so long as the system can safely remain in the target once it arrives, reaching the target by any time guarantees that the system can be in the target at the final time.
The following result is then immediate from Theorem \ref{th:main} and Proposition \ref{prop:tube_set}.
\begin{theorem}\label{th:main-2}
    Under Assumption \ref{as:rcitarget}, the value function $V_0:\mathbb{R}^n \times (-\infty,0] \to \mathbb{R}$ is a VB-CBF for system \eqref{eq:sys} on the reach-avoid tube $\mathcal{RA}(\mathcal{T},\mathcal{F},\cdot)$. 
\end{theorem}
\AB{The above theorem implies that the value function $V_0$ can be computed to construct a VB-CBF on the reach-avoid tube (rather than the reach-avoid set as before), provided that Assumption~\ref{as:rcitarget} holds.}
\vspace{-1mm}
\subsection{Safety Filter} \label{sec:safety-filter}
Given a VB-CBF $h_v$ for system \eqref{eq:sys} and corresponding extended class $\mathcal{K}$ function $\alpha$, we consider for $t \le 0$ and $x \in \mathcal{C}_v(t)$ the feasible control set
\begin{align}
    \Omega(x,t):=\{u \in \mathcal{U} \mid& \exists \psi:\mathbb{R}^n \times (-\infty,0) \to \mathbb{R} \text{ continuously differentiable s.t. } h_v - \psi \text{ has a local}\nonumber\\
    &\text{minimum at } (x,\!t),
     \frac{\partial}{\partial t} \psi(x,\!t) \!+\! \min_{d \in \mathcal{D}} \! \nabla_x  \psi(x,\! t)^\top\! f(x,\! u,\! d)  \!\geq\! - \!\alpha(h_v(x, t))\! \}. \!
\end{align}
We use the feasible control set to design a safety filter that prioritizes system safety and liveness, while also addressing a secondary objective. Specifically, given a nominal control law $u_{\mathrm{nom}}: \mathbb{R}^n \times (-\infty,0] \to \mathcal{U}$ that might violate safety or prevent safely achieving the target set, whenever $\Omega(x,t)$ is non-empty, we can minimally adjust the nominal control input to a feasible one using the following optimization problem
\begin{equation}\label{eq:qp_vb_cbf}
\begin{aligned}
u^*(x,t) =&\arg \min _{u \in \mathcal{U}}\left\|u - u_{\text {nom }}(x)\right\|_2^2 \\
&\text { s.t. }  \frac{\partial}{\partial t} \psi(x,t) +\min_{d \in \mathcal{D}} \nabla_x \psi(x, t)^\top f (x, u, d)  \geq - \alpha(h_v(x, t)),
\end{aligned}
\end{equation}
where $\psi$ is chosen as above.
When the dynamics of the system are also control and disturbance-affine, the optimization problem can then be expressed as a quadratic program (QP) that can be solved quickly and precisely online.
\begin{remark}
    \DCH{ When computing $u^*(x,t)$ online, note that if} $h_v$ is differentiable at $(x,t)$, it follows from Lemma 1.7 of \cite{bardi} that one can simply solve the optimization problem \eqref{eq:qp_vb_cbf} with $h_v$ in place of $\psi$.
    \DCH{Additionally,} since $V_0$ is Lipschitz (c.f.,~\cite{fisac2015reach}), $V_0$ is differentiable almost everywhere by Rademacher's Theorem.
    It then follows from Theorems \ref{th:main} and \ref{th:main-2}, that for $h_v = V_0$ we have that $\Omega(x,t)$ is non-empty almost everywhere in $\mathcal{RA}_0(\mathcal{T},\mathcal{F},\cdot)$ (and $\mathcal{RA}(\mathcal{T},\mathcal{F},\cdot)$ as well under Assumption \ref{as:rcitarget}).
    \DCH{These observations justify that in practice, if one uses the value function $V_0$ computed via HJR as a VB-CBF, they may compute $u^*(x,t)$ by solving the optimization problem \eqref{eq:qp_vb_cbf} with both $\psi$ and $h_v$ replaced with $V_0$.
    In this case, there will be some feasible control action at almost every $x$ and $t$ in the reach-avoid tube.} 
\end{remark}
\vspace{-1mm}
\section{Numerical Experiments}
\begin{figure}[t!]
    \centering
    \includegraphics[width=0.65\textwidth]{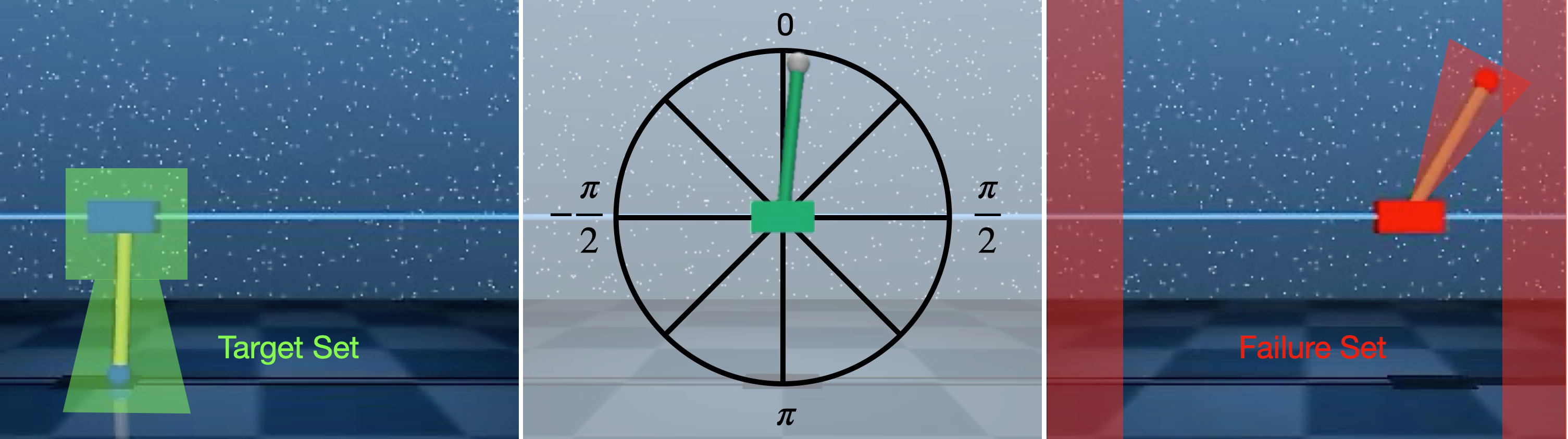}
    \caption{Our modified cartpole environment based on the Deepmind control suite. 
    The green region in the left figure depicts the target region that we desire to reach, the center image depicts the swung up position that yields the highest reward and the red regions in the right image depict unsafe regions in the state space. 
    }
    \label{fig:env}
\end{figure}

\subsection{Simulation Setup}
\textbf{Modified Cartpole Environment Setup: }We demonstrate the \hfill\hfill\hfill\hfill
\begin{wrapfigure}{r}{0.35\textwidth} 
    \centering
    \vspace{-1cm}
    \captionsetup{justification=justified,singlelinecheck=false}
    \includegraphics[width=0.3\textwidth]{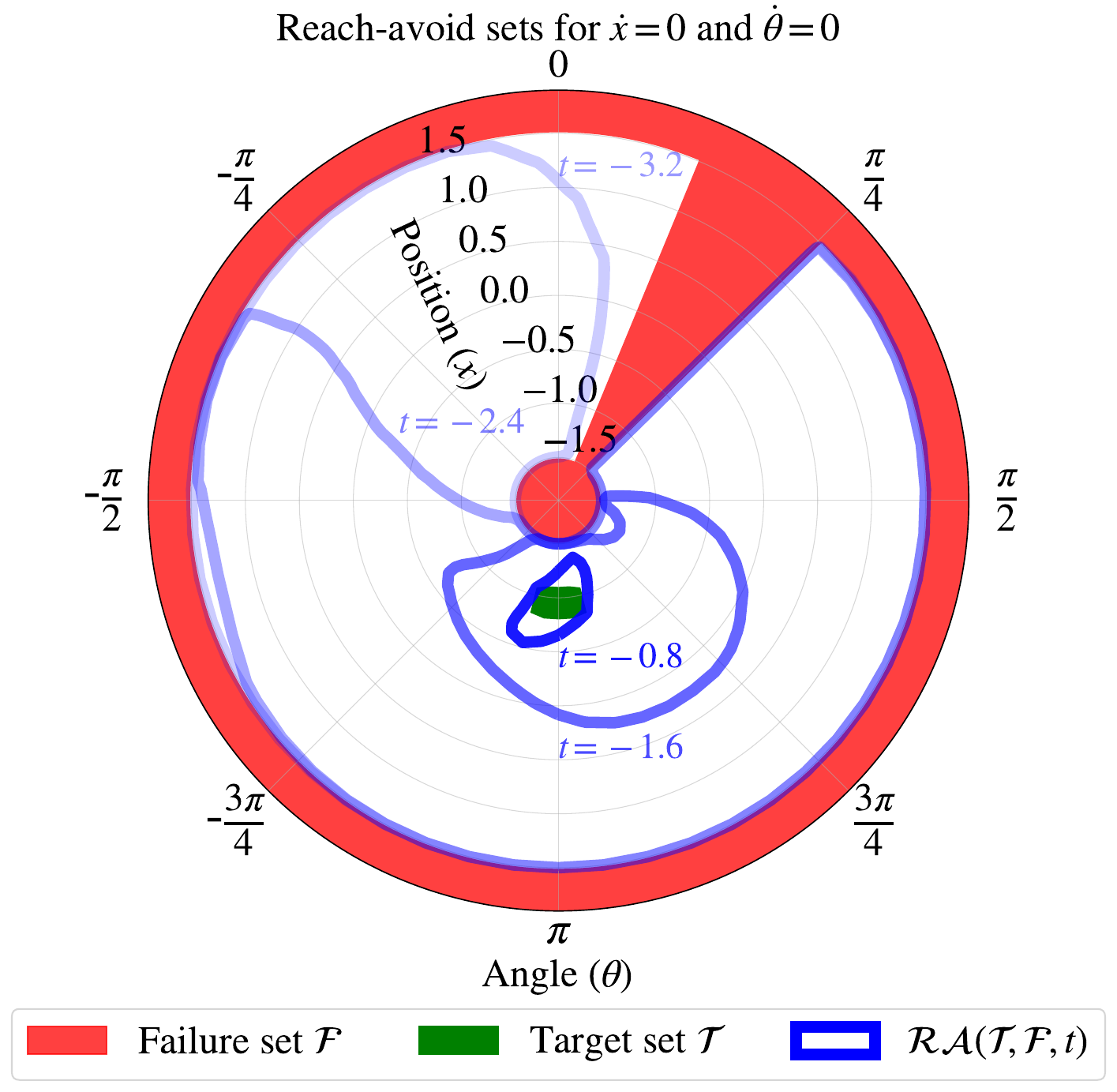} 
    \caption{Reach-avoid computation for the modified cartpole environment. The reach-avoid tubes $\mathcal{RA}(\mathcal{T}, \mathcal{F}, t)$ indicate states from which the cartpole can reach the target $\mathcal{T}$ at times $t=\{-0.8, -1.6, -2.4, -3.2\}$ seconds while avoiding the failure set $\mathcal{F}$.}
    \label{fig:sets}
\end{wrapfigure} effectiveness of our approach by considering the cartpole swing-up task~\citep{tunyasuvunakool2020}, modified such that all the mass is placed at the end of the pole, and damping and friction are removed. The state space of the cartpole system is given by $z =[x, \theta, \dot{x}, \dot{\theta}]^\top \in \mathbb{R}^4$, where $x$ is the position of the cart, $\theta$ is the pendulum angle, and $\dot{x}$ and $\dot{\theta}$ are the velocity of the cart and the angular velocity of the pendulum. The state space of environment is constrained by $x \in [-1.8, 1.8]$. For our setting, we define position $x$ as unsafe for $x < -1.5$ or $x>1.5$, while the angular position $\theta$ is unsafe for $\pi/8 < \theta < \pi/4$, visualized in Figures~\ref{fig:env} and~\ref{fig:sets} in red. These constraints limit the safe swing-up behavior of the system to the clockwise direction. These unsafe regions are not explicitly modeled in the simulator, thus safety violations do not directly impact performance. We consider the target set as a region where the robot must reset at the end of each episode. For this setup, the target region is specified as
\begin{equation*}
   \!  \mathcal{T} \!=  \!  \left\{ \!  z  \! \mid  \! x \!  \in \!  [-1.1, -0.8], \theta \in [-\pi - 0.25, \pi + 0.25], \dot{x} \in [-0.1, 0.1], \dot{\theta} \in [-0.25, 0.25] \right\}. \! 
\end{equation*} 
This target set is visualized in Figures~\ref{fig:env} and~\ref{fig:sets} in green. Terminating an episode in this target region is not straightforward. 
For our system to efficiently achieve this state, it must actively dissipate energy while simultaneously moving to the desired position, balancing precise position control with energy-reducing dynamics.
The zero-level sets of the reach-avoid tube $\mathcal{RA}(\mathcal{T}, \mathcal{F}, t)$ are shown over various times $t$ in Figure~\ref{fig:sets}.
\newline
\textbf{RL Policy Modification:} We consider the widely used RL method soft actor-critic (SAC) \citep{Haarnoja2018SoftAA}, which utilizes a maximum entropy RL objective to learn both a stochastic policy and a value function. We build on~\cite{pytorch_sac} implementation. However, to ensure safety during training, we run the output of SAC through a safety filter. By combining SAC with a reach-avoid VB-CBF safety filter we aim to learn to efficiently swing-up the cartpole to maximize reward, while avoiding safety violations and returning to the reset position (i.e. target set $\mathcal{T}$).
\begin{table}[t!]
    \centering
    \begin{tabular}{@{}lcccc@{}}
         & SAC & SAC-CBF & \textbf{SAC-RACBF} & \textbf{SAC-RACBF-noreset} \\ \midrule
        $\%$ Unsafe Trajectories & $99.437$ & $0.563$ &\textbf{$0.070$} & \textbf{$0.000$} 
    \end{tabular}
    \captionsetup{justification=centering}
    \caption{Average (over 5 seeds) percentage of unsafe trajectories during training. Our methods are bolded.}
    \label{tab:perc}
\end{table}

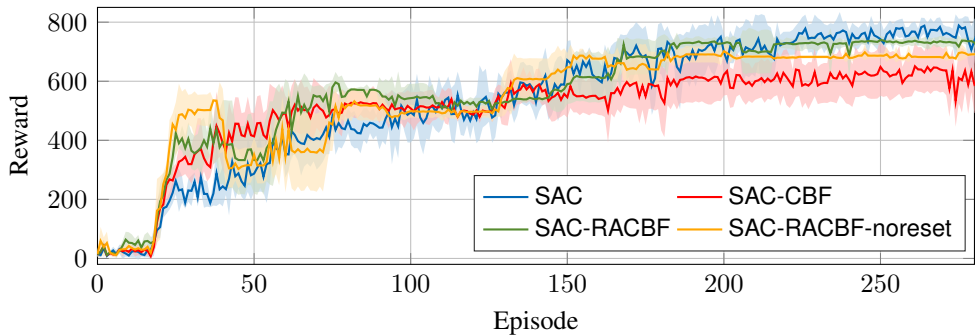
\begin{figure}[t!]
    \centering
    \definecolor{nicegreen}{RGB}{0,200,0} 
	\pgfplotsset{width=10\columnwidth /10, compat = 1.13, 
		height = 4cm, grid= major, 
		legend cell align = left, ticklabel style = {font=\small\sffamily},
		every axis label/.append style={font=\small},
		legend style = {font=\small, at={(0.99,0.05)}, anchor=south east},title style={yshift=7pt, font = \small}}

\centering
\begin{tikzpicture}
\node[name=plotvirtual] at (0,0) {\begin{tikzpicture}
\begin{axis}[
name=plot1,
grid=both, 
height = 4cm, width =\columnwidth-1cm,
xmin=0, xmax = 280,
ymin=-20, ymax=850,
xtick={0,50,100,150,200,250},
ylabel={Reward},
xlabel={Episode },
legend columns=2,
legend style={font=\footnotesize\sffamily , legend cell align=left, align=left, draw=white!15!black},
]

\addplot[thick, azrablue] table[x=t, y=x] {sac.txt}; 
\addplot [name path=lower, fill=none, draw=none, forget plot]  table[x=t, y=c1] {sac.txt}; 
\addplot [name path=upper, fill=none, draw=none, forget plot] table[x=t, y=c2] {sac.txt}; 
\addplot[azrablue!30, opacity=0.6, forget plot] fill between[of=lower and upper];

\addplot[thick, red] table[x=t, y=x] {sac-cbf.txt}; 
\addplot [name path=lower, fill=none, draw=none, forget plot]  table[x=t, y=c1] {sac-cbf.txt}; 
\addplot [name path=upper, fill=none, draw=none, forget plot] table[x=t, y=c2] {sac-cbf.txt}; 
\addplot[red!30, opacity=0.6, forget plot] fill between[of=lower and upper];

\addplot[thick, gr] table[x=t, y=x] {sac-racbf.txt}; 
\addplot [name path=lower, fill=none, draw=none, forget plot]  table[x=t, y=c1] {sac-racbf.txt};  
\addplot [name path=upper, fill=none, draw=none, forget plot] table[x=t, y=c2] {sac-racbf.txt}; 
\addplot[gr!30, opacity=0.6, forget plot] fill between[of=lower and upper];

\addplot[thick, or] table[x=t, y=x] {sac-racbf-noreset.txt}; 
\addplot [name path=lower, fill=none, draw=none, forget plot]  table[x=t, y=c1] {sac-racbf-noreset.txt}; 
\addplot [name path=upper, fill=none, draw=none, forget plot] table[x=t, y=c2] {sac-racbf-noreset.txt}; 
\addplot[or!30, opacity=0.6, forget plot] fill between[of=lower and upper];

\addlegendentry{$\text{SAC}$}
\addlegendentry{$\text{SAC-CBF}$}
\addlegendentry{$\text{SAC-RACBF}$}
\addlegendentry{$\text{SAC-RACBF-noreset}$}
\end{axis}
\end{tikzpicture}};
\end{tikzpicture}
\vspace{-5mm}
    \caption{
    The mean and standard deviation of episode rewards over 5 different seeds on the given baselines (SAC, SAC-CBF) and proposed methods (SAC-RACBF and SAC-RACBF-noreset). The proposed methods achieve reward performances closely comparable to the baselines, demonstrating how the reach-avoid VB-CBF safety filter ensures safety, guarantees a safe return to the desired target set, and minimally impacts the learning process of the SAC agent. To ensure a fair comparison, we compare the reward returned over the first $10$ seconds of the $15$ second episode, as our method requires returning to the target set, unlike the baselines.}  \label{fig:rewards}
\end{figure}

\subsection{Results and Analysis}\label{sec:results}
To evaluate the effectiveness of our approach, we compare the following methods:
\begin{enumerate} \vspace{-2mm}
\item \textbf{(Ours)} SAC with reach-avoid VB-CBF (SAC-RACBF): Our method integrates a reach-avoid VB-CBF with SAC, combining safety with a reachability objective to guide the agent. \vspace{-3mm} 
\item \textbf{(Ours)} SAC with reach-avoid VB-CBF without resetting (SAC-RACBF-noreset): A modification of the above, where the next episode starts from the final state of the previous episode.\vspace{-3mm}
\item Standard SAC (SAC): The standard SAC algorithm, which does not account for safety constraints or the objective of returning to the start state, serves as our performance baseline.\vspace{-3mm}
\item SAC with CBF (SAC-CBF): A variant of SAC enhanced with an CBF-like safety filter. Inspired by the approach in \cite{ChoiLeeEtAl2021}, we construct a CBF-like safety filter following the same methodology as SAC-RACBF, but relying on the value function obtained from solving the avoid-only problem.\vspace{-3mm}
\end{enumerate}
All methods and their results are averaged over 5 seeds. The reward graph over the course of training is  illustrated in the Figure~\ref{fig:rewards}. These showcase the rewards accrued over the first $10$ seconds of the $15$ second episode, as we aim to terminate the episode at a safe reset state in the target set (which has a low reward for the swing-up objective). The safety-filter enhanced SAC agents agents accumulate rewards comparable to standard SAC, demonstrating that the safety filter based on the proposed method is minimally invasive and allows the cartpole agent to efficiently learn how to swing up. Importantly, unlike standard SAC, our methods preserve safety during learning. 
Table~\ref{tab:perc} shows how our methods are better at maintaining safety throughout the training, while standard SAC leads to unsafe trajectories in over $99\%$ of the episodes. 
SAC-CBF provides similar levels of safety to our approaches, however, it does not provide guarantees on returning to the target set and thus terminates each episode far away from the desired safe return region. 
Figure~\ref{fig:final_delta_states} shows the minimum difference between the current state and target set, denoted by $\Delta$, for each dimension.
The graphs illustrate how our methods closely align with the desired target region throughout all episodes, while the baselines remain widely distributed across the state space.

In summary, the results in Table~\ref{tab:perc}, Figures~\ref{fig:rewards} and~\ref{fig:final_delta_states} demonstrate how our proposed method has minimal detrimental effect on performance and maintains safety, while enabling a key feature for real world RL, namely terminating an episode in a desired safe reset region.
Furthermore, the performance of the proposed method SAC-RACBF-noreset highlights its capability for hands-off RL, enabling efficient and cost-effective robot training in real world scenarios. We acknowledge that a policy trained with the reach-avoid-based safety filter will require the continued use of this safety filter during online deployment. However, the modifications to the safety filter for a policy can be incorporated in a manner that eliminates their necessity post-training~\citep{Cheng2019EndtoEndSR}.
\begin{figure}[t!]
     \centering
     \vspace{-0.3cm}
    {\text{$\Delta$ Final state from Target Region over episodes}} 
    \pgfplotsset{
        width=0.5\textwidth,
        height=4cm,
        grid=none,
        legend columns=4, 
        legend style={font=\small, at={(1.9,-1.9)}, anchor=south east, legend cell align=left},
        ylabel style={font=\small},
        xlabel style={font=\small},
        ticklabel style={font=\small}
    }

    \begin{tikzpicture}
        \begin{axis}[
            height=3.5cm, width =(\columnwidth)/2,
            xmin=0, xmax=280,
            ymin=-1, ymax=3,
            ylabel={$\Delta x$},
            xticklabels={,,}
        ]
            \addplot[thick, azrablue] table[x=t, y=x] {sac_x.txt}; 
            \addplot [name path=lower, fill=none, draw=none, forget plot]  table[x=t, y=c1] {sac_x.txt}; 
            \addplot [name path=upper, fill=none, draw=none, forget plot] table[x=t, y=c2] {sac_x.txt}; 
            \addplot[azrablue!30, opacity=0.6, forget plot] fill between[of=lower and upper];

            \addplot[thick, red] table[x=t, y=x] {sac-cbf_x.txt}; 
            \addplot [name path=lower, fill=none, draw=none, forget plot]  table[x=t, y=c1] {sac-cbf_x.txt}; 
            \addplot [name path=upper, fill=none, draw=none, forget plot] table[x=t, y=c2] {sac-cbf_x.txt}; 
            \addplot[red!30, opacity=0.6, forget plot] fill between[of=lower and upper];

            \addplot[thick, gr] table[x=t, y=x] {sac-racbf_x.txt}; 
            \addplot [name path=lower, fill=none, draw=none, forget plot]  table[x=t, y=c1] {sac-racbf_x.txt};  
            \addplot [name path=upper, fill=none, draw=none, forget plot] table[x=t, y=c2] {sac-racbf_x.txt}; 
            \addplot[gr!30, opacity=0.6, forget plot] fill between[of=lower and upper];

            \addplot[thick, or] table[x=t, y=x] {sac-racbf-noreset_x.txt}; 
            \addplot [name path=lower, fill=none, draw=none, forget plot]  table[x=t, y=c1] {sac-racbf-noreset_x.txt}; 
            \addplot [name path=upper, fill=none, draw=none, forget plot] table[x=t, y=c2] {sac-racbf-noreset_x.txt}; 
             \addplot[or!30, opacity=0.6, forget plot] fill between[of=lower and upper];

            \addlegendentry{$\text{SAC}$}
            \addlegendentry{$\text{SAC-CBF}$}
            \addlegendentry{$\text{SAC-RACBF}$}
            \addlegendentry{$\text{SAC-RACBF-noreset}$}
        \end{axis}
        
        \begin{axis}[
            height=3.5cm, width =(\columnwidth)/2,
            at={(7.5cm,0cm)},
            xmin=0, xmax=280,
            ymin=-4.5, ymax=4.5,
            ylabel={$\Delta \theta$},
            xticklabels={,,}
        ]
              \addplot[thick, azrablue] table[x=t, y=x] {sac_theta.txt}; 
            \addplot [name path=lower, fill=none, draw=none, forget plot]  table[x=t, y=c1] {sac_theta.txt}; 
            \addplot [name path=upper, fill=none, draw=none, forget plot] table[x=t, y=c2] {sac_theta.txt}; 
            \addplot[azrablue!30, opacity=0.6, forget plot] fill between[of=lower and upper];

            \addplot[thick, red] table[x=t, y=x] {sac-cbf_theta.txt}; 
            \addplot [name path=lower, fill=none, draw=none, forget plot]  table[x=t, y=c1] {sac-cbf_theta.txt}; 
            \addplot [name path=upper, fill=none, draw=none, forget plot] table[x=t, y=c2] {sac-cbf_theta.txt}; 
            \addplot[red!30, opacity=0.6, forget plot] fill between[of=lower and upper];

            \addplot[thick, gr] table[x=t, y=x] {sac-racbf_theta.txt}; 
            \addplot [name path=lower, fill=none, draw=none, forget plot]  table[x=t, y=c1] {sac-racbf_theta.txt};  
            \addplot [name path=upper, fill=none, draw=none, forget plot] table[x=t, y=c2] {sac-racbf_theta.txt}; 
            \addplot[gr!30, opacity=0.6, forget plot] fill between[of=lower and upper];

            \addplot[thick, or] table[x=t, y=x] {sac-racbf-noreset_theta.txt}; 
            \addplot [name path=lower, fill=none, draw=none, forget plot]  table[x=t, y=c1] {sac-racbf-noreset_theta.txt}; 
            \addplot [name path=upper, fill=none, draw=none, forget plot] table[x=t, y=c2] {sac-racbf-noreset_theta.txt}; 
             \addplot[or!30, opacity=0.6, forget plot] fill between[of=lower and upper];
             
        \end{axis}
        
        \begin{axis}[
            height=3.5cm, width =(\columnwidth)/2,
            at={(0cm,-2cm)},
            xmin=0, xmax=280,
            ymin=-3, ymax=3,
            xlabel={$\text{Episode}$},
            ylabel={$\Delta \dot{x}$},
        ]
            \addplot[thick, azrablue] table[x=t, y=x] {sac_xdot.txt}; 
            \addplot [name path=lower, fill=none, draw=none, forget plot]  table[x=t, y=c1] {sac_xdot.txt}; 
            \addplot [name path=upper, fill=none, draw=none, forget plot] table[x=t, y=c2] {sac_xdot.txt}; 
            \addplot[azrablue!30, opacity=0.6, forget plot] fill between[of=lower and upper];

            \addplot[thick, red] table[x=t, y=x] {sac-cbf_xdot.txt}; 
            \addplot [name path=lower, fill=none, draw=none, forget plot]  table[x=t, y=c1] {sac-cbf_xdot.txt}; 
            \addplot [name path=upper, fill=none, draw=none, forget plot] table[x=t, y=c2] {sac-cbf_xdot.txt}; 
            \addplot[red!30, opacity=0.6, forget plot] fill between[of=lower and upper];

            \addplot[thick, gr] table[x=t, y=x] {sac-racbf_xdot.txt}; 
            \addplot [name path=lower, fill=none, draw=none, forget plot]  table[x=t, y=c1] {sac-racbf_xdot.txt};  
            \addplot [name path=upper, fill=none, draw=none, forget plot] table[x=t, y=c2] {sac-racbf_xdot.txt}; 
            \addplot[gr!30, opacity=0.6, forget plot] fill between[of=lower and upper];

            \addplot[thick, or] table[x=t, y=x] {sac-racbf-noreset_xdot.txt}; 
            \addplot [name path=lower, fill=none, draw=none, forget plot]  table[x=t, y=c1] {sac-racbf-noreset_xdot.txt}; 
            \addplot [name path=upper, fill=none, draw=none, forget plot] table[x=t, y=c2] {sac-racbf-noreset_xdot.txt}; 
             \addplot[or!30, opacity=0.6, forget plot] fill between[of=lower and upper];
        \end{axis}
        
        \begin{axis}[
            height=3.5cm, width =(\columnwidth)/2,
            at={(7.5cm,-2cm)},
            xmin=0, xmax=280,
            ymin=-22, ymax=22,
            xlabel={$\text{Episode}$},
            ylabel={$\Delta \dot{\theta}$},
        ]
            \addplot[thick, azrablue] table[x=t, y=x] {sac_thetadot.txt}; 
            \addplot [name path=lower, fill=none, draw=none, forget plot]  table[x=t, y=c1] {sac_thetadot.txt}; 
            \addplot [name path=upper, fill=none, draw=none, forget plot] table[x=t, y=c2] {sac_thetadot.txt}; 
            \addplot[azrablue!30, opacity=0.6, forget plot] fill between[of=lower and upper];

            \addplot[thick, red] table[x=t, y=x] {sac-cbf_thetadot.txt}; 
            \addplot [name path=lower, fill=none, draw=none, forget plot]  table[x=t, y=c1] {sac-cbf_thetadot.txt}; 
            \addplot [name path=upper, fill=none, draw=none, forget plot] table[x=t, y=c2] {sac-cbf_thetadot.txt}; 
            \addplot[red!30, opacity=0.6, forget plot] fill between[of=lower and upper];

            \addplot[thick, gr] table[x=t, y=x] {sac-racbf_thetadot.txt}; 
            \addplot [name path=lower, fill=none, draw=none, forget plot]  table[x=t, y=c1] {sac-racbf_thetadot.txt};  
            \addplot [name path=upper, fill=none, draw=none, forget plot] table[x=t, y=c2] {sac-racbf_thetadot.txt}; 
            \addplot[gr!30, opacity=0.6, forget plot] fill between[of=lower and upper];

            \addplot[thick, or] table[x=t, y=x] {sac-racbf-noreset_thetadot.txt}; 
            \addplot [name path=lower, fill=none, draw=none, forget plot]  table[x=t, y=c1] {sac-racbf-noreset_thetadot.txt}; 
            \addplot [name path=upper, fill=none, draw=none, forget plot] table[x=t, y=c2] {sac-racbf-noreset_thetadot.txt}; 
             \addplot[or!30, opacity=0.6, forget plot] fill between[of=lower and upper];
        \end{axis}
    \end{tikzpicture}
\vspace{-2mm}
    \caption{The mean and standard deviation of the difference between the final state and desired target region, averaged over the training episodes of $5$ seeds. The proposed methods SAC-RACBF and SAC-RACBF-reset remain at $\Delta$ close to $0$, indicating they successfully reach the target region, while the baselines (SAC and SAC-CBF) are distributed throughout the state space.}
    \vspace{-0.5cm}
    \label{fig:final_delta_states}
\end{figure}
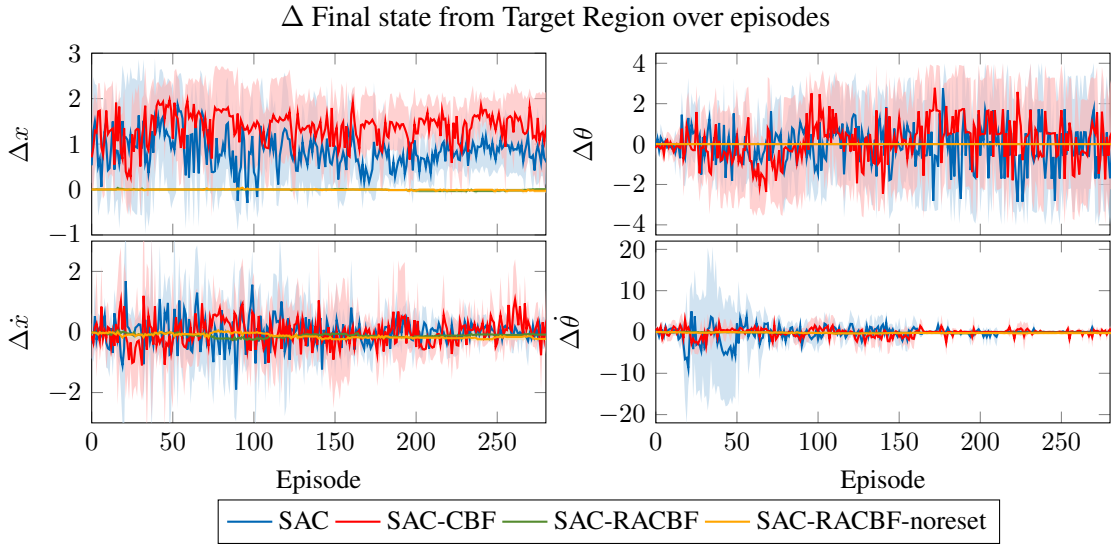

\vspace{-2mm}
\section{Conclusion} \label{sec:conclusion}
In this paper, we present a novel concept of reach-avoid VB-CBF that integrates CBFs with the HJR reach-avoid set. When combined with a safety filter, this approach prioritizes invariance with respect to the time-varying safe set while maintaining robustness to control and disturbance bounds. Moreover, we show that the Hamilton-Jacobi reach-avoid set is a reach-avoid VB-CBF. We motivate our approach with the promise of safe, efficient, hands-off RL for training robots in the real world. The effectiveness of our method is demonstrated through safe training and resetting in a cart-pole environment.
\AB{For future work, we plan to carry out a comparison of our approach with other methods and leverage DeepReach to approximate the reach-avoid value function for high-dimensional and partially unknown dynamics using online updates from a conservative initial guess.}
\section*{Acknowledgements}
This research was supported by ONR grants N00014-22-1-2292 and N00014-23-1-2353, the Naval Innovation, Science, Engineering Center (NISEC) at UC San Diego under grant N00014-23-1-2831, and NIBIB of the National Institutes of Health under grant T32EB009380.
\bibliography{references}

\end{document}